\documentclass[12pt]{article}
\usepackage{amsthm,amsfonts, amsmath}
\newtheorem{theorem}{Theorem}%[section]

\newtheorem{remark}{Remark}%[section]
\newtheorem{proposition}{Proposition}
\newtheorem{lemma}{Lemma}%[section]

\newtheorem{step}{Step}
\def\geqslant {\ge}
\def\leqslant {\le}
\def\bq{\begin{eqnarray}}
\def\eq{\end{eqnarray}}
\def\bqq{\begin{eqnarray*}}
\def\eqq{\end{eqnarray*}}
\def\nn{\nonumber}

\def\B {\mathcal{B}}

\def\R {\mathbb{R}}
\def\C {\mathbb{C}}

\def\A{\mathcal{A}}

\def\K{\mathcal{K}}

\def\d{\text{d}}

\def\eps{\varepsilon}

\title{New bounds on the maximum ionization of atoms}
\author{Phan Th\`anh Nam\\\\
\small Department of Mathematical Sciences, University of Copenhagen,\\
\small Universitetsparken 5, 2100 Copenhagen, Denmark. E-mail: ptnam@math.ku.dk}
\begin{document}

\date{{}}
\maketitle

\begin{abstract} We prove that the maximum number $N_c$ of non-relativistic electrons that a nucleus of charge $Z$ can bind is less than $1.22\/Z+3Z^{1/3}$. This improves Lieb's upper bound $N_c<2Z+1$ [{\it Phys. Rev. A} {\bf 29}, 3018-3028 (1984)] when $Z\ge 6$. Our method also applies to non-relativistic atoms in magnetic field and to pseudo-relativistic atoms. We show that in these cases, under appropriate conditions, $\limsup_{Z\to \infty}N_c/Z \le 1.22$.


\end{abstract}

\section{Introduction}
Let us consider an atom with a classical nucleus of charge $Z$ and $N$ non-relativistic quantum electrons. The nucleus is fixed at the origin and the $N$-electron system is described by the Hamiltonian
\[
H_{N,Z} = \sum\limits_{i = 1}^N {\left( { - \frac{1}{2}\Delta _i  - \frac{Z}
{{|x_i |}}} \right)}  + \sum\limits_{1 \leqslant i < j \leqslant N} {\frac{1}
{{|x_i  - x_j |}}} 
\]
acting on the antisymmetric space $\mathop  \bigwedge \limits_{i = 1}^N (L^2 (\mathbb{R}^3 ) \otimes \mathbb{C}^2)$. The nuclear charge $Z$ is allowed to be any positive number, although it is an integer in the physical case. 

The ground state energy of $N$ electrons is the bottom of the spectrum of $H_{N,Z}$, 
\[
E(N,Z) = \inf {\text{spec}}~ H_{N,Z}=\mathop {\inf }\limits_{||\psi ||_{L^2}  = 1} (\psi ,H_{N,Z} \psi ).
\]
We say that $N$ electrons can be {\it bound} if $E(N,Z)<E(N-1,Z)$, namely one cannot remove any electron without paying some positive energy. Due to the HVZ theorem (see e.g. \cite{Te09}, Theorem 11.2) which states that
$$\text{ess spec}~H_{N,Z}=[E(N - 1),\infty),$$
one always has $E(N)\le E(N-1)$. Moreover the binding inequality $E(N,Z)<E(N-1,Z)$ means that $E(N,Z)$ is an isolated eigenvalue of $H_{N,Z}$. Zhislin (1960) \cite{Zh60} show that binding occurs provided that $N<Z+1$.  

Of our interest is the maximum number $N_c=N_c(Z)$ of electrons that can be bound. It is a long standing open problem, sometimes referred to as the {\it ionization conjecture} (see e.g. \cite{Li84,So91,So03,LS09}), that $N_c\le Z+1$, or maybe $N_c\le Z+2$. Note that $N_c\ge Z$ due to Zhislin's result. We now briefly present the status of the conjecture, and we refer to \cite{LS09} (Chap. 12) for a pedagogical introduction to this problem.

It was first proved by Ruskai (1982) \cite{Ru82} and Sigal (1982, 1984) \cite{Si82,Si84} that $N_c$ is not too large. In fact,  Ruskai \cite{Ru82} showed that $N_c=O(Z^{6/5})$ as $Z\to \infty$ and Sigal \cite{Si84} showed that $N_c\le 18Z$ and $\limsup_{Z\to \infty}N_c/Z\le 2$. Then Lieb (1984) \cite{Li84} gave a very simple and elegant proof that $N_c<2Z+1$ for all $Z>0$. Lieb's upper bound settles the conjecture for hydrogen but it is around twice of the conjectured bound for large $Z$. 

For large atoms, the asymptotic neutrality $\lim_{Z\to \infty}N_c/Z=1$ was first proved by Lieb, Sigal, Simon and Thirring (1988) \cite{LSST88}. Later, it was improved to $N_c\le Z+O(Z^{5/7})$ by Seco, Sigal and Solovej (1990) \cite{SSS90} and by Fefferman and Seco (1990) \cite{FS90}. The bound $N_c\le Z+\text{const}$, for some $Z$-independent constant, is still unknown, although it holds true for some important approximation models such as Thomas-Fermi and related theories \cite{Li81,BL85} and Hartree-Fock theory \cite{So91, So03}.  

In spite of the asymptotic neutrality, Lieb's upper bound $N_c<2Z+1$ \cite{Li84} is still the best one for realistic atoms (corresponding to the range $1\le Z\le 118$ in the current periodic table). The purpose in this work is to find an improved bound for all $Z>0$. As in \cite{Li84}, we do not need the binding inequality; more precisely, that $E(N,Z)$ is an eigenvalue of $H_{N,Z}$ is sufficient for our analysis. One of our main result is the following. 
   
\begin{theorem}[Bound on maximum ionization of non-relativistic atoms] \label{thm:main-theorem} Let $Z>0$ (not necessarily an integer). If $E(N,Z)$ is an eigenvalue of $H_{N,Z}$, then either $N=1$ or 
$$N<1.22\/Z+3Z^{1/3}.$$
The factor $1.22$ can be replaced by $\beta^{-1}$ with $\beta$ being defined by (\ref{eq:def-beta}) below.
\end{theorem}

\begin{remark} The bound $1.22\/Z+3Z^{1/3}$ is less than Lieb's bound $2Z+1$ when $Z\ge 6$.  
\end{remark}

\begin{remark} While Lieb's result holds true for both fermions and bosons, our result only holds for fermions (in fact, our method works also for the bosonic case but it yields an estimate worse than Lieb's). Note that the ionization conjecture only concerns fermions since for bosonic atoms it was shown that $\lim_{Z\to \infty}N_c/Z =t_c \approx 1.21$ by Benguria and Lieb \cite{BL83} and Solovej \cite{So90} (the numerical value $1.21$ is taken from \cite{Ba84}). In our proof below, we use Pauli's exclusion principle in Lemma \ref{le:lemma1}. More precisely, we use the fact that in a fermionic atom the average distance from the electrons to the nucleus of charge $Z$ is (at least) of order $Z^{-1/3}$. In contrast, the corresponding distance in the bosonic atoms is of order $Z^{-1}$.
\end{remark}

\begin{remark} Although Lieb's method \cite{Li84} can be generalized to molecules, we have not yet been able to adapt our method to this case. 
\end{remark}

Our method also applies to other models such as non-relativistic atoms in magnetic fields and relativistic atoms, and we shall discuss these extensions later. In the rest of the introduction let us outline the proof of Theorem \ref{thm:main-theorem}. As a first step we get the following bound. 

\begin{lemma}\label{le:Z>alphaN-N} If $E(N,Z)$ is an eigenvalue of $H_{N,Z}$ then we have
$$
\alpha _{N}(N-1)< Z(1+0.68~N^{-2/3}),
$$
where
\bq \label{eq:def-alphaN}
\alpha _N : = \mathop {\inf }\limits_{x_1 ,...,x_N  \in \mathbb{R}^3 } \frac{{\sum\limits_{1 \leqslant i < j \leqslant N} {\frac{{|x_i |^2  + |x_j |^2 }}
{{|x_i  - x_j |}}} }}
{{(N - 1)\sum\limits_{i = 1}^N {|x_i |} }}.
\eq
\end{lemma} 

This result is proved by modifying Lieb's proof: in \cite{Li84} Lieb multiplied the eigenvalue equation $(H_{N,Z}-E(N,Z))\Psi_{N,Z}=0$ by $|x_N|\overline{\Psi}_{N,Z}$. We instead multiply by $x_N^2\overline{\Psi}_{N,Z}$ and employ the Lieb-Thirring inequality to control error terms. 

Roughly speaking, the number $\alpha_N^{-1}$ yields an upper bound on $N/Z$. This bound improves previous results since $\alpha_N$ is bigger than $1/2$ (one can see that $\alpha_2=1/2$ and $\alpha_N\ge \sqrt{5}/4 =0.559...$ when $N\ge 3$). Although we do not know the exact value of $\alpha_N$, it is possible to derive some effective estimates. We may think of $\alpha_N$ as the lowest energy of $N$ classical particles acting on $\R^3$ via the potential $V(x,y)=\frac{x^2+y^2}{|x-y|}$, under some normalizing condition. It is natural to believe that if $N$ becomes large, then $\alpha_N$ converges to the statistical limit
\bq \label{eq:def-beta}
\beta : = \inf \left\{ {\frac{{\iint\limits_{\mathbb{R}^3  \times \mathbb{R}^3 } {\frac{{x^2+y^2}}
{{2|x - y|}} \operatorname{d\rho} (x)\operatorname{d\rho} (y)}}}
{{\int\limits_{\mathbb{R}^3 } {|x|  \operatorname{d\rho} (x)} } }:\rho~\text{a probability measure on}~\R^3} \right\}.
\eq
Results of this form in bounded domain have already appeared in \cite{MS82}. Indeed, we can show that $\alpha_N$ actually converges to $\beta$ and provide an explicit estimate on the convergence rate. Theorem \ref{thm:main-theorem} essentially follows by inserting the lower bound on $\alpha_N$ in Proposition \ref{pro:alphaN} below into the inequality in Lemma \ref{le:Z>alphaN-N}.

\begin{proposition} \label{pro:alphaN} The sequence $\{\alpha_N\}_{N=2}^\infty$ is increasing and for any $N\ge 2$ we have
\[
\beta  \geqslant \alpha _N  \geqslant \frac{N}{N-1}[\beta  - 3(\beta /6)^{1/3} N^{ - 2/3}],
\]
with $\beta$ being defined by (\ref{eq:def-beta}). Moreover, $\beta\in [0.8218,0.8705)$.
\end{proposition}

\begin{remark} We do not know the exact numerical value of $\beta$, but our bound that $\beta\in [0.82188,0.8705)$ is already rather precise. There is of course still room for improvement. 
\end{remark}

The article is organized as follows. We shall prove Theorem \ref{thm:main-theorem} in Section 2. Then we discuss some possible extensions of our method in Section 3. Proposition \ref{pro:alphaN} is of independent interest and we defer its proof to Section 4.  
\text{}\\\\
{\bf Acknowledgments.} I am grateful to Mathieu Lewin and Jan Philip Solovej for very helpful discussions, and to Rupert L. Frank and Elliott H. Lieb for pointing out the lower bound (\ref{eq:LT-lower-bound-ground-state-energy}) which improves the constants in Lemma \ref{le:lemma1} and Theorem \ref{thm:main-theorem}. I thank the referee for his constructive comments. This work was done when I was a visiting student at D\'epartement de Math\'ematiques, Universit\'e de Cergy-Pontoise (France), and I wish to thank the people there for the warm hospitality.

\section{Proof of Theorem \ref{thm:main-theorem}: the new bound}

\subsection{Lieb's method} In order to make our argument transparent we start by quickly recalling the proof of Lieb \cite{Li84}.  Assume that $E(N,Z)$ is an eigenvalue of $H_{N,Z}$ corresponding to some normalized eigenfunction $\Psi_N$. Multiplying the Schr\"odinger equation 
\bq \label{eq:Schrodinger-equation}
(H_{N,Z}-E(N,Z))\Psi_{N,Z}=0
\eq
by $|x_N|\overline{\Psi}_{N,Z}$  and then integrating, one gets
\bq \label{eq:0=xuHu=T1+T2+T3}
  0 &=& \left< {|x_N |\Psi_{N,Z} ,(H_{N,Z}  - E(N,Z) )\Psi_{N,Z} } \right>\nn \hfill \\
   &=& \left< {|x_N |\Psi_{N,Z} ,(H_{N - 1,Z}  - E(N,Z) )\Psi_{N,Z} } \right>+ \frac{1}{2}\left< {|x_N |\Psi_{N,Z} , - \Delta _N \Psi_{N,Z} } \right>\nn\hfill\\
&~& + \left< {\Psi_{N,Z} ,\left[ { - Z + \sum\limits_{i = 1}^{N - 1} {\frac{{|x_N |}}
{{|x_i  - x_N |}}} } \right]\Psi_{N,Z} } \right>.
\eq

The first term in the right hand side of (\ref{eq:0=xuHu=T1+T2+T3}) is non-negative since $H_{N-1,Z}\ge E(N-1,Z)\ge E(N,Z)$ (in the space of $N-1$ particles $x_1$,...,$x_{N-1}$). The second term is also non-negative due to the inequality
\bq \label{eq:xfAf>=0}
\operatorname{Re} \left< {|x|f, - \Delta f} \right>\ge 0~~\text{for all}~f\in H^1(\R^3).
\eq
 
Thus the third term in (\ref{eq:0=xuHu=T1+T2+T3}) must be non-positive. Using the symmetry of $|\Psi_{N,Z}|^2$ (which holds true for both fermions and bosons) we can rewrite it as 
\[
\left\langle {\Psi _{N,Z} ,\left( { - Z + \frac{1}
{N}\sum\limits_{1 \leqslant i < j \leqslant N} {\frac{{|x_i | + |x_j |}}
{{|x_i  - x_j |}}} } \right)\Psi _{N,Z} } \right\rangle \le 0.
\]
It follows from the triangle inequality that
\bq \label{eq:triangular-ineq}
\frac{1}
{{N(N - 1)}}\sum\limits_{1 \leqslant i < j \leqslant N} {\frac{{|x_i | + |x_j |}}
{{|x_i  - x_j |}}}  \geqslant \frac{1}
{2}.
\eq
Hence we obtain $-ZN+\frac{N(N-1)}{2}<0$, namely $N<2Z+1$. The inequality is strict since the triangle inequality is strict almost everywhere in $(\R^{3})^N$. Note that the lower bound $1/2$ in (\ref{eq:triangular-ineq}) is sharp (when $|x_i|\ll |x_j|$ if $i<j$). 

\subsection{Proof of Lemma \ref{le:Z>alphaN-N}}

Instead of multiplying the equation (\ref{eq:Schrodinger-equation}) by $|x_N|\overline{\Psi}_{N,Z}$, we now  multiply by $x_N^2\overline{\Psi}_{N,Z}$ and integrate. We obtain
\bqq
  0 &=& \left\langle {x_N^2 \Psi _{N,Z} ,(H_{N - 1,Z}  - E_{N,Z} )\Psi _{N,Z} } \right\rangle  + \frac{1}
{2}\left\langle {x_N^2 \Psi _{N,Z} , - \Delta_N \Psi _{N,Z} } \right\rangle  \hfill \\
   &~&~~+ \left\langle {\Psi _{N,Z} ,\left( { - Z|x_N | + \frac{1}{N}\sum\limits_{1 \leqslant i < j \leqslant N} {\frac{{x_i^2  + x_j^2 }}
{{|x_i  - x_j |}}} } \right)\Psi _{N,Z} } \right\rangle  \hfill \\
   &\geqslant& \frac{1}
{2}\left\langle {x_N^2 \Psi _{N,Z} , - \Delta_N \Psi _{N,Z} } \right\rangle  + \left\langle {\Psi _{N,Z} ,\left( { - Z + \alpha _N (N - 1)} \right)|x_N |\Psi _{N,Z} } \right\rangle.
\eqq
Recall that $\alpha_N$ is defined in (\ref{eq:def-alphaN}). This implies that
\bq \label{eq:0=xxuHu>=}
\alpha _N (N - 1) \leqslant Z - \frac{1}
{2}\left\langle {x_N^2 \Psi _{N,Z} , - \Delta_N \Psi _{N,Z} } \right\rangle \left\langle {\Psi _{N,Z} ,|x_N |\Psi _{N,Z} } \right\rangle ^{ - 1} .
\eq

As we will see, the main advantage of our method is that the number $\alpha_N$ is bigger than $1/2$ when $N\ge 3$. However, we do not have an inequality similar to (\ref{eq:xfAf>=0}) with $|x|$ replaced by $x^2$. In fact, for all $f\in H^1(\R^3)$, applying the identity
\bq \label{eq:identity-xxT}
\operatorname{Re} \left\langle {\varphi f, - \Delta f} \right\rangle  = \left\langle {\varphi ^{1/2} f,\left( { - \Delta  - \left| {\frac{{\nabla \varphi }}
{{2\varphi }}} \right|^2 } \right)\varphi ^{1/2} f} \right\rangle 
\eq
to $\varphi(x)=x^2$ we find that  
\bq \label{eq:xxDelta>=-3/4}
\text{Re}\left\langle { x^2 f, - \Delta f} \right\rangle  = \left\langle {f, (|x|(-\Delta)|x|-1) f} \right\rangle   \ge   - \frac{3}{4}\left\langle {f,f} \right\rangle 
\eq
by Hardy's inequality, $-3/4$ being the sharp constant. 

Our observation is that we may still control the second term in the right hand side of (\ref{eq:0=xxuHu>=}) since $\left\langle {\Psi _{N,Z} ,|x_N |\Psi _{N,Z} } \right\rangle^{-1}$ is small (in comparison with $Z$) . In fact, $\left\langle {\Psi _{N,Z} ,|x_N |\Psi _{N,Z} } \right\rangle$ can be understood as the average distance from $N$ electrons to the nucleus, which is well-known to be (at least) of order $Z^{-1/3}$. We have the following explicit bound.

\begin{lemma}\label{le:lemma1} If $\Psi_{N,Z}$ is a ground state of $H_{N,Z}$ then   
$$\left\langle {\Psi _{N,Z} ,|x_N |\Psi _{N,Z} } \right\rangle >0.553~Z^{-1}N^{2/3}.$$
\end{lemma}

It follows from (\ref{eq:xxDelta>=-3/4}) and Lemma \ref{le:lemma1} that
\[
\frac{1}
{2}\left\langle {x_N^2 \Psi _{N,Z} , - \Delta \Psi _{N,Z} } \right\rangle \left\langle {\Psi _{N,Z} ,|x_N |\Psi _{N,Z} } \right\rangle ^{ - 1}  >  - 0.68~ZN^{ - 2/3} .
\]
Substituting the latter estimate into (\ref{eq:0=xxuHu>=}) we obtain the inequality in Lemma \ref{le:Z>alphaN-N}. We now provide the

\begin{proof}[Proof of Lemma \ref{le:lemma1}] The following proof essentially follows from \cite{LS09} ( p. 132). Note that
$$\left\langle {\Psi _{N,Z} ,|x_N |\Psi _{N,Z} } \right\rangle=\frac{1}{N}\int_{\R^3}{|x|\rho_{\Psi_{N,Z}}(x)\d x}$$
where the density $\rho_{\Psi_{N,Z}}$ of $\Psi_{N,Z}$ is defined by
\[
\rho _{\Psi _{N,Z} } (x): = N\sum_{\sigma_1=1,2}...\sum_{\sigma_N=1,2} \int\limits_{\mathbb{R}^{3(N - 1)} } {|\Psi_{N,Z} (x,\sigma_1;x_2,\sigma_2;...;x_N,\sigma_N)|^2\d x_2...\d x_N } .
\]

By solving the Bohr atom as in \cite{Li76} (after eq. (40) p. 560) one has the lower bound on the ground state energy
\bq \label{eq:LT-lower-bound-ground-state-energy}
E(N,Z) \geqslant \left\langle {\Psi _{N,Z} ,\sum\limits_{i = 1}^N {\left( { - \frac{1}
{2}\Delta _i  - \frac{Z}
{{|x_i|}}} \right)\Psi _{N,Z} } } \right\rangle  \geqslant  - AZ^2 N^{1/3} 
\eq
where $A=(3^{1/3}/2)2^{2/3}$. Moreover, one has the Lieb-Thirring kinetic energy inequality \cite{LT75} 
\bq \label{eq:LT-lower-bound-kinetic}
\K_{\Psi_{N,Z}}:=\frac{1}{2}\sum\limits_{i = 1}^N {\left\langle {\Psi_{N,Z} , - \Delta _i \Psi_{N,Z}} \right\rangle }\ge K\int\limits_{\mathbb{R}^3 } {\rho _{\Psi_{N,Z}}  (x)^{5/3} dx},
\eq
where $K = 2^{-2/3}(3/10)\left( {2/(5L)} \right)^{2/3}$ with $L=(\pi3^{3/2} 5)^{-1}= 0.01225...$ (this constant $L$ is taken from \cite{DLL08}). Since $E(N,Z)=-\K_{\Psi_{N,Z}}$ by the Virial Theorem, we get from (\ref{eq:LT-lower-bound-ground-state-energy}) and (\ref{eq:LT-lower-bound-kinetic}) that
\bq \label{eq:rho5/3-upper-bound}
\int\limits_{\mathbb{R}^3 } {\rho _{\Psi_{N,Z}}  (x)^{5/3} \d x}  \leqslant K^{-1}AZ^2 N^{1/3}.
\eq

On the other hand, we have the following inequality introduced by Lieb (\cite{Li76}, p. 563)
\[
\left( {\int\limits_{\mathbb{R}^3 } {\varphi (x)^{5/3} \d x} } \right)^{p/2} \left( {\int\limits_{\mathbb{R}^3 } {|x|^p \varphi (x)\d x} } \right) \geqslant C_p \left( {\int\limits_{\mathbb{R}^3 } {\varphi (x)\d x} } \right)^{1 + 5p/6} 
\]
for any nonnegative measurable function $\varphi (x)$, with the sharp constant $C_p$ being attained for $\varphi(x)=(1-|x|^p)^{3/2}_+$. In particular, applying this inequality to $\varphi(x)=\rho_{\Psi_{N,Z}}(x)$ and $p=1$, we get
\bq \label{eq:|x|rho-lower-bound}
\left( {\int\limits_{\mathbb{R}^3 } {\rho _{\Psi_{N,Z} } (x)^{5/3} \d x} } \right)^{ 1/2} \int\limits_{\mathbb{R}^3 } {|x|\rho _{\Psi_{N,Z} } (x)\d x}  \geqslant C_1 N^{11/6} 
\eq
where $C_1=\pi^{-1/3}2^{-1}3^{5/3}5^{5/6}7^{1/3}11^{-3/2}= 0.4271...$  Combining (\ref{eq:rho5/3-upper-bound}) and (\ref{eq:|x|rho-lower-bound}) we obtain the desired inequality.  
\end{proof}

\subsection{Proof of Theorem \ref{thm:main-theorem}}

Let us admit Proposition \ref{pro:alphaN} for the moment and derive Theorem \ref{thm:main-theorem}. Lemma \ref{le:Z>alphaN-N} and Proposition \ref{pro:alphaN} together yield a lower bound on $Z$ in terms of $N$, 
\bq \label{eq:main-theorem-lower-bound-Z}
\frac{{N(\beta  - 3(\beta /6)^{1/3} N^{ - 2/3} )}}
{{1 + 0.68~N^{ - 2/3} }} < Z.
\eq

It is just an elementary calculation to translate (\ref{eq:main-theorem-lower-bound-Z}) into an upper bound on $N$ in terms of $Z$. If $\min\{N,Z\}<3$, then $\max\{2,\beta^{-1}Z+3Z^{1/3}\}>2Z+1$ (since $\beta<0.8705$), and hence our bound follows from Lieb's bound $N<2Z+1$. If $\min\{N,Z\}\ge 3$, then $\beta  - 3(\beta /6)^{1/3} N^{ - 2/3}>0$ and Lieb's bound implies that $N/Z< 2+Z^{-1}\le  7/3$. Thus the desired result follows from (\ref{eq:main-theorem-lower-bound-Z}) and the following technical lemma whose proof is provided in the Appendix. 

\begin{lemma}\label{le:lemma3} For $Z>0$, $N>0$, $N/Z<7/3$ and $\beta\ge 0.8218$ one has  
$$
\beta^{-1}Z+3Z^{1/3} > \min \left\{ {N,Z \frac{1+0.68~N^{-2/3}}{\beta-3(\beta /6)^{1/3}N^{-2/3}}} \right\}.
$$
\end{lemma}

\section{Some possible extensions}

\subsection{Atoms in magnetic fields}

In this section, we consider the ionization problem with the presence of a magnetic field. The system is now described by the Hamiltonian
\[
H_{N,Z,\A}  = \sum\limits_{i = 1}^N {\left( {T_\A^{(i)}  - \frac{Z}
{{|x_i |}}} \right)}  + \sum\limits_{1 \leqslant i < j \leqslant N} {\frac{1}
{{|x_i  - x_j |}}} 
\]
acting on the fermionic space $\mathop  \bigwedge \limits^N (L^2 (\mathbb{R}^3 ) \otimes \mathbb{C}^2)$. The kinetic operator is the Pauli operator
\[
T_\A  =\left| {\sigma  \cdot ( - i\nabla  + \A(x))} \right|^2  =( - i\nabla  + \A(x))^2 +\sigma  \cdot \B ,
\]
where $\A$ is the magnetic potential, $\B=\text{curl}(\A)$ is the magnetic field and $\sigma=(\sigma^1,\sigma^2,\sigma^3)$ are the Pauli matrices
\[
\sigma ^1  = \left( {\begin{array}{*{20}c}
   0 & 1  \\
   1 & 0  \\

 \end{array} } \right),\sigma ^2  = \left( {\begin{array}{*{20}c}
   0 & { - i}  \\
   i & 0  \\

 \end{array} } \right),\sigma ^3  = \left( {\begin{array}{*{20}c}
   1 & 0  \\
   0 & { - 1}  \\

 \end{array} } \right).
\]

For simplicity we shall always assume that $\A\in L^4_{\text{loc}}(\R^3,\R^3)$, $\nabla  \cdot \A\in L^2_{\text{loc}}(\R^3)$ and $|\B|\in L^{3/2}(\R^3)+L^{\infty}(\R^3)$. Under these assumptions, it is well known that $( - i\nabla  + \A(x))^2$ is essentially self-adjoint on $L^2(\R^3)$ with the core $C^\infty_c(\R^3)$ \cite{LS81}, and $|\B| + Z/|x|$ is infinitesimally bounded with respect to $( - i\nabla  + \A(x))^2$ (see e.g. \cite{Se01}). In particular, the ground state energy  
$$E(N,Z,\B)=\text{inf spec}~H_{N,Z,\A}$$
is finite. We shall also assume that $N\mapsto E(N,Z,\B)$ is non-increasing (for example, this is the case if $\B=(0,0,B)$ is a constant magnetic field \cite{LSY94I}). Note that the ground state energy depends on $\A$ only through $\B$ by {\it gauge invariance} (see e.g. \cite{LS09} p. 21).   

Of our interest is the maximum number $N_c$ such that $E(N_c,Z,\B)$ is an eigenvalue of $H_{N,Z,\A}$. 
%If the coupling $\sigma \cdot \B$ in $T_\A$ is ignored then both Lieb's upper bound $N_c<2Z+1$ and our bound in Theorem \ref{thm:main-theorem} still hold (in fact, Hardy's inequality (\ref{eq:xxDelta>=-3/4}) and the Lieb-Thirring we have used in the proof of Lemma \ref{le:lemma1} do not change with $-\Delta$ replaced by $( - i\nabla  + \A(x))^2$). However, it is unknown if Lieb's bound $N_c<2Z+1$, which is independent of $\B$, holds for the general Pauli operator $T_\A$ or not. 
Seiringer \cite{Se01} showed in 2001 that
\bq \label{eq:Seiringer-1}
N_c  < 2Z + 1 +\frac{1}{2} \frac{{E(N_c ,Z,\B) - E(N_c ,kZ,\B)}}
{{N_c Z(k - 1)}}
\eq
for all $k>1$. In the homogeneous case, $\B=(0,0,B)$, his bound yields
\bq \label{eq:Seiringer-2}
N_c  < 2Z + 1 + C_1 Z^{1/3}  + C_2 Z\min \left\{ {(B/Z^3 )^{2/5} ,1+|\ln (B/Z^3 )|^{2} } \right\}.
\eq
In particular, in the semiclassical regime $\lim_{Z\to \infty}(B/Z^3)=0$, Seiringer's bound implies that 
\[
\mathop {\lim \sup }\limits_{Z \to \infty } \frac{{N_c }}
{Z} \leqslant 2.
\]
In contrast, it was shown by Lieb, Solovej and Yngvason (1994) \cite{LSY94I} that if $\lim_{Z\to \infty}(B/Z^3)=\infty$, then 
\[
\mathop {\lim \inf }\limits_{Z \to \infty } \frac{{{N_c }}}
{Z} \geqslant 2.
\]

We shall improve these upper bounds using the method in the previous section. Our result in this section is as follows.

\begin{theorem}[Bounds on maximum ionization of atoms in magnetic fields]\label{thm:bound-magnetic-1} Assume that $\A\in L^4_{\text{loc}}(\R^3,\R^3)$, $\nabla  \cdot \A\in L^2_{\text{loc}}(\R^3)$ and $|\B|=| {\rm curl}(\A)|\in L^{3/2}(\R^3)+L^{\infty}(\R^3)$. For any $Z>0$, denote by $N_c=N_c(Z)$ the maximum number such that $E(N_c,Z,\B)$ is an eigenvalue of $H_{N,Z,\A}$. Then we have, for every $k>1$, 
\bq \label{eq:bound-magnetic-1}
N_c  < (1.22Z + 3Z^{1/3} )\left( {1 +\frac{{E(N_c ,Z,\B) - E(N_c ,kZ,\B)}}
{{N_c Z^2(k - 1)}}} \right).
\eq

If $\B=(0,0,B)$ is a constant magnetic field, then 
\bqq \label{eq:bound-magnetic-2}
N_c < (1.22~Z + 3Z^{1/3} )  \left( {1 + 11.8~Z^{ - 2/3} + \min \left\{ {0.42\left( {B/Z^3 } \right)^{2/5} ,C(1 + |\ln (B/Z^3 )|^2 )} \right\}} \right)
\eqq
for some universal constant $C$ (independent of $Z$ and $B$). In particular, if $\lim_{Z\to \infty}(B/Z^3)=0$, then
\[
\mathop {\limsup}\limits_{Z \to \infty } \frac{{N_c }}
{Z} \leqslant 1.22.
\]
The number $1.22$ in all bounds can be replaced by $\beta^{-1}$ with $\beta$ being defined by (\ref{eq:def-beta}).
\end{theorem}

\begin{proof} Assume that $\Psi_{N,Z,\A}$ is a ground state of $H_{N,Z,\A}$. Following the proof of Lemma \ref{le:Z>alphaN-N}, we have 
\bq \label{eq:magnetic-0=xxuHu>=}
\alpha _N (N - 1) \leqslant Z - \left\langle {x_N^2 \Psi _{N,Z,\A} , T_\A \Psi _{N,Z,\A} } \right\rangle \left\langle {\Psi _{N,Z,\A} ,|x_N |\Psi _{N,Z,\A} } \right\rangle ^{ - 1} ,
\eq
which is the analogue of (\ref{eq:0=xxuHu>=}). 

We may assume that $N\ge \beta^{-1}Z+3Z^{-2/3}$ (otherwise we are done). In this case the left hand side of (\ref{eq:magnetic-0=xxuHu>=}) can be bound by
\bq \label{eq:magnetic-alphaN}
{\alpha_N(N-1)}>\frac{N}{\beta ^{ - 1}  + 3Z^{ - 2/3}}.
\eq 
This estimate follows from the lower bound on $\alpha_N$ in Proposition \ref{pro:alphaN} and the following technical lemma whose proof is provided in the Appendix.

\begin{lemma}\label{le:lemma4} For $Z>0$, $N\in \mathbb{N}$, $N\ge \beta^{-1}Z+3Z^{-2/3}$ and $\beta\ge 0.8218$, one has 
\[
(\beta  - 3(\beta/6)^{1/3} N^{ - 2/3} )(\beta ^{ - 1}  + 3Z^{ - 2/3} ) > 1.
\]
\end{lemma}

The second term in the right hand side of (\ref{eq:magnetic-0=xxuHu>=}) can be bound in the same way as in \cite{Se01}. More precisely, using (\ref{eq:identity-xxT}) with $-\Delta$ replaced by $T_\A\ge 0$, one has 
\bq \label{eq:magnetic-lower-bound-xxT}
\left\langle {x_N^2 \Psi_{N,Z,\A},T_\A \Psi_{N,Z,\A}} \right\rangle  = \left\langle {\Psi_{N,Z,\A},\left( {|x_N|T_\A|x_N|  -1} \right)\Psi_{N,Z,\A}} \right\rangle  \geqslant  - 1.
\eq
On the other hand, for every $k>1$, 
\bq \label{eq:estimate-E-E}
{{\left\langle {\Psi _{N,Z,\A} ,|x_N |\Psi _{N,Z,\A} } \right\rangle }}^{-1} &\leqslant&  \left\langle {\Psi _{N,Z,\A} ,|x_N |^{ - 1} \Psi _{N,Z,\A} } \right\rangle \nn\hfill\\
&=&\frac{{\left\langle {\Psi _{N,Z,\A} ,H_{N,Z,\A} \Psi _{N,Z,\A} } \right\rangle  - \left\langle {\Psi _{N,Z,\A} ,H_{N,kZ,\A} \Psi _{N,Z,\A} } \right\rangle }}
{{NZ(k - 1)}} \nn\hfill\\
&\leqslant& \frac{{E(N,Z,\B) - E(N,kZ,\B)}}
{{NZ(k - 1)}}
\eq
since $\Psi _{N,Z,\A}$ is a ground state of $H_{N,Z,\A}$. Then (\ref{eq:bound-magnetic-1}) follows by substituting (\ref{eq:magnetic-alphaN}), (\ref{eq:magnetic-lower-bound-xxT}) and (\ref{eq:estimate-E-E}) into (\ref{eq:magnetic-0=xxuHu>=}). 

Now assume that $\B=(0,0,B)$ is a constant magnetic field. It follows from \cite{LSY94II} (Theorems 2.4, 2.5) that if $N\ge Z/2$, then the ground state energy $E(N,Z,B):=E(N,Z,\B)$ can be bounded from below by  
\bq \label{eq:magnetic-lower-bound-ground-state-energy}
E(N,Z,B) \ge  - NZ^2 \left( {18.7Z^{ - 2/3} +  \min \left\{ {0.95\left( {B/Z^3 } \right)^{2/5} ,C\left( {1 + |\ln (B/Z^3 )|^2 } \right)} \right\}} \right)
\eq

for some universal constant $C$ (independent of $N$, $Z$ and $B$). (It is obtained when applying (2.27), (2.26), (2.29) in \cite{LSY94II} to the cases: $B< Z^{4/3}$, $B\ge Z^{4/3}$, $B\gg Z^3$, respectively.) 

We can choose $k=2$ in (\ref{eq:bound-magnetic-1}). Then the desired bound follows by using the upper bound $E(N,Z,B)\le 0$ and the lower bound on $E(N,2Z,B)$ derived from  (\ref{eq:magnetic-lower-bound-ground-state-energy}). 
\end{proof}

\begin{remark}\label{remark:remark5} We may also consider the Hamiltonian $H_{N,Z,\A}$ on the bosonic space $\mathop \otimes \limits_{\text{sym}}^N (L^2 (\mathbb{R}^3 ) \otimes \mathbb{C}^q)$, where $q$ is a spin number. In this case the inequality (\ref{eq:bound-magnetic-1}) still holds true. Moreover, if $\B=(0,0,B)$ is a constant magnetic field, then using the estimate \cite{Se01} (p. 1948) 
\[
E(N,Z,B) = NZ^2 E(1,1,B/Z^2 ) \geqslant  - \frac{1}
{4}NZ^2 \min \left\{ {1 + 4B/Z^2 ,C|\ln (B/Z^2 )|^2 } \right\}
\]
we get from (\ref{eq:bound-magnetic-1}) that
\[
N_c  < (\beta ^{ - 1} Z + 3Z^{1/3} )\left( {1 + \min \left\{ {1+4B/Z^2 ,C_2 |\ln (B/Z^2 )|^2 } \right\} } \right).
\]
In particular, if $\lim_{Z\to \infty} (B/Z^2)= 0$, then our bound yields 
$$\limsup_{Z\to \infty} \frac{N_c}{Z} \le 2\beta^{-1} \le 2.44.$$
It slightly improves the bosonic bound in \cite{Se01}, which gives $\limsup_{Z\to \infty} (N_c/Z) \le 2.5.$  
\end{remark}

% \begin{remark} As shown in \cite{MJ05}, one can make a slight improvement on our bounds by using the Hardy-type inequality $T_{\A} \ge (d_{\B}/4)|x|^{-2}$ instead of $T_{\A}\ge 0$ in (\ref{eq:magnetic-lower-bound-xxT}), for some $0<d_\B\le 1$. It allows us to include a factor $(1-d_\B)$ in front of the term involving to $E(N,Z,\B)-E(N,kZ,\B)$ in (\ref{eq:bound-magnetic-1}). 
% \end{remark}

\subsection{Pseudo-relativistic atoms}

In this section we consider the pseudo-relativistic Hamiltonian
\[
H_{N,Z}^{\text{rel}}  = \sum\limits_{i = 1}^N {\left( {\alpha^{-1}(\sqrt { - \Delta _i  +\alpha^{-1}}-\alpha^{-1})- \frac{{Z}}
{{|x_i |}}} \right)}  + \sum\limits_{1 \leqslant i < j \leqslant N} {\frac{1}
{{|x_i  - x_j |}}}
\]
acting on the fermionic space $\mathop  \bigwedge \limits^N (L^2 (\mathbb{R}^3 ) \otimes \mathbb{C}^2)$. Here $\alpha>0$ is the {\it fine-structure constant}. It is well known that the ground state energy $E^{\text{rel}}(N,Z):=\inf \text{spec}~H_{N,Z}^{\text{rel}}$ is finite if and only if $Z\alpha \le 2/\pi$ (see e.g. \cite{LS09}). The physical value is $\alpha=e^2 /(\hbar c)\approx 1/137$ and hence $Z<87.22$. However, we allow $\alpha$ to be any positive number. 

As in the previous dicussions, we are also interested in the maximum number $N_c$ such that the ground state energy $E^{\text{rel}}(N_c,Z)$ is an eigenvalue of $H_{N_c,Z}^{\text{rel}}$. Note that Lieb's bound $N_c<2Z+1$ still holds in this case. In fact, due to a technical gap the original proof of Lieb in \cite{Li84} works properly only when $Z\alpha<1/2$. However, it is posible to fill this gap to obtain the bound up to $Z\alpha<2/\pi$ \cite{DSS}. On the other hand, to our knowledge, no result about the asymptotic behavior of $N_c/Z$ is available for the pseudo-relativistic model, although within pseudo-relativistic Hartree-Fock theory it was recently shown by Dall'Acqua and Solovej (2010) \cite{DS10} that $N_c\le Z+\text{const}$.

Our result in this section is the following.

\begin{theorem}[Bound on maximum ionization of pseudo-relativistic atoms] Let $Z>0$ such that $Z\alpha \le \kappa <2/\pi$. If $E^{{\rm rel}}(N,Z)$ is an eigenvalue of $H_{N,Z}^{{\rm rel}}$, then either $N=1$ or  
\[
N  < 1.22Z + C_\kappa Z^{1/3}
\]
for some constant $C_\kappa$ depending only on $\kappa$. The number $1.22$ can be replaced by $\beta^{-1}$ with $\beta$ being defined by (\ref{eq:def-beta}).
\end{theorem}

\begin{proof} Assume that $\Psi_{N,Z}^{\text{rel}}$ is a ground state of $H_{N,Z}^{\text{rel}}$. As an analogue of (\ref{eq:0=xxuHu>=}) we get
\bq \label{eq:relativistic-0=xxuHu>=}
\alpha _N (N - 1) \leqslant Z - \left\langle {x_N^2\Psi_{N,Z}^{\text{rel}} , \alpha^{-1}(\sqrt{-\Delta_N+\alpha^{-1}}-\alpha^{-1}) \Psi_{N,Z}^{\text{rel}}} \right\rangle \left\langle {\Psi_{N,Z}^{\text{rel}} ,|x_N |\Psi_{N,Z}^{\text{rel}} } \right\rangle ^{ - 1} .
\eq

The left hand side of (\ref{eq:relativistic-0=xxuHu>=}) can be bounded using (\ref{eq:magnetic-alphaN}). Turning to the right hand side of (\ref{eq:relativistic-0=xxuHu>=}), we first show that for any function $f: \R^3\to \C$ smooth enough
\bq \label{eq:relativistic-lower-bound-xxT}
 \operatorname{Re} \left\langle {x^2 ,\alpha^{-1}\left( {\sqrt { - \Delta + \alpha^{-1}}  -\alpha^{-1}} \right)f} \right\rangle _{L^2 (\mathbb{R}^3 ,\operatorname{dx})}\ge - \frac{3}
{8}\left\langle {f,f} \right\rangle .
\eq
It suffices to show (\ref{eq:relativistic-lower-bound-xxT}) for $\alpha=1$ (the general case follows by scaling). Using the Fourier transform $\widehat f(p) := \int\limits_{\mathbb{R}^3 } {e^{ - i2\pi  p \cdot x} f(x)\d x}$ and applying (\ref{eq:identity-xxT}) to $$\varphi (p):=\sqrt{(2\pi p)^2+1}-1=\frac{(2\pi p)^2}{\sqrt{(2\pi p)^2+1}+1}$$
we find that   
\bqq
&~&\operatorname{Re} \left\langle {x^2 f,\left[ {\sqrt { - \Delta _x  + 1}  - 1} \right]f} \right\rangle _{L^2 (\mathbb{R}^3 ,\d x)}   \hfill\\
&=& (2\pi )^{ - 2} \operatorname{Re} \left\langle { - \Delta _p \widehat f,\varphi \widehat f} \right\rangle _{L^2 (\mathbb{R}^3 ,\d p)} \hfill\\
&=&(2\pi )^{ - 2}   \left\langle {\varphi ^{1/2} \widehat f,\left( { - \Delta_p  - \left| {\frac{{\nabla \varphi }}
{{2\varphi }}} \right|^2 } \right)\varphi ^{1/2} \widehat f} \right\rangle  \hfill \\
  &=&(2\pi )^{ - 2}  \left\langle {\varphi ^{1/2} \widehat f,\left( { - \Delta_p  - \frac{{(\sqrt {(2\pi p)^2  + 1}  + 1)^2 }}
{{4p^2 ((2\pi p)^2  + 1)}}} \right)\varphi ^{1/2} \widehat f} \right\rangle.
\eqq
Then it follows from Hardy's inequality $-\Delta_p\ge 1/(4p^2)$ that 
\bqq
&~& \operatorname{Re} \left\langle {x^2 f,\left[ {\sqrt { - \Delta _x  + 1}  - 1} \right]f} \right\rangle _{L^2 (\mathbb{R}^3 ,\d x)} \hfill\\
&\geqslant&  - \left\langle {\widehat f,\frac{{2\sqrt {(2\pi p)^2  + 1}  + 1}}
{{4((2\pi p)^2  + 1)(\sqrt {(2\pi p)^2  + 1}  + 1)}}\widehat f} \right\rangle_{L^2(\R^3,\d p)}  \hfill \\
   &\ge& - \frac{3}
{8}\left\langle {\widehat f,\widehat f} \right\rangle  =  - \frac{3}
{8}\left\langle {f,f} \right\rangle .
\eqq

The term $\left\langle {\Psi ,|x_N |\Psi } \right\rangle ^{ - 1}$ can be estimated similarly to (\ref{eq:estimate-E-E}), namely
$$\left\langle {\Psi ,|x_N |\Psi } \right\rangle ^{ - 1}\le \left\langle {\Psi ,|x_N |^{-1}\Psi } \right\rangle \le \frac{E(N,Z)-E(N,kZ)}{NZ(k-1)}$$
for every $k>1$ such that $kZ\alpha<2/\pi$. It is well known that $0\ge E(N,Z)\ge -C_\kappa Z^{7/3}$ provided that $Z\alpha\le \kappa$. In fact, it was shown by S\o rensen \cite{Sorensen07} that, in the limit $Z\to \infty$ (and $Z\alpha=\kappa$ fixed), the leading order of the ground-state energy $E(N,Z)$ is given by the Thomas-Fermi theory which is of order $Z^{7/3}$. Thus we can conclude that 
\bq  \label{eq:relativistic-lower-bound-<|xN|>}
\left\langle {\Psi ,|x_N |\Psi } \right\rangle ^{ - 1}\le C_\kappa Z^{-2/3}.
\eq
The desired result follows from (\ref{eq:relativistic-lower-bound-xxT}), (\ref{eq:relativistic-lower-bound-<|xN|>}), (\ref{eq:relativistic-0=xxuHu>=}) and (\ref{eq:magnetic-alphaN}). 
\end{proof}

\section{Proof of Proposition \ref{pro:alphaN}: Analysis of $\alpha_N$}

This section is devoted to the proof of Proposition \ref{pro:alphaN}. For the reader's convenience, we split the proof into several steps. Recall that $\alpha_N$ and $\beta$ are defined in (\ref{eq:def-alphaN}) and (\ref{eq:def-beta}), respectively.

\begin{step}\label{st:thm3-st1} The sequence $\alpha_N$ is increasing in $N$ and it converges to $\beta$ as $N\to \infty$. 
\end{step} 

\begin{proof} The fact that  $\alpha_N$ is increasing is shown as follows: for every $x_1,...,x_N\in \R^3$ we have
\bqq
  \sum\limits_{1 \leqslant i < j \leqslant N} {\frac{{x_i^2  + x_j^2 }}
{{|x_i  - x_j |}}} &=& \sum\limits_{k = 1}^N {\left( {\frac{1}
{{(N - 2)}}\sum\limits_{i < j;i \ne k,j \ne k} {\frac{{x_i^2  + x_j^2 }}
{{|x_i  - x_j |}}} } \right)}  \hfill \\
  &\geqslant& \sum\limits_{k = 1}^N {\left( {\alpha _{N - 1} \sum\limits_{i \ne k} {|x_i |} } \right)}  = \alpha _{N - 1} (N - 1)\sum\limits_{i = 1}^N {|x_i |},
\eqq
where we have used the definition of $\alpha_{N-1}$. This implies that $\alpha_{N-1}\le \alpha_N$. 

We shall show that $\alpha_N$ converges to $\beta$. We start with the upper bound $\alpha_N\le \beta$. Let $\rho$  be an arbitrary probability measure on $\R^3$. Then 
\bqq
  \iint\limits_{\mathbb{R}^3  \times \mathbb{R}^3 } {\frac{{x^2+y^2 }}
{{2|x - y|}}\operatorname{d\rho} (x)\operatorname{d\rho} (y)}  &=&\int\limits_{\mathbb{R}^{3N} } {\frac{1}
{{N(N - 1)}}\sum\limits_{1\le i < j \le N} {\frac{{x_i^2  + x_j^2 }}
{{|x_i  - x_j |}}} \operatorname{d\rho} (x_1 )...\operatorname{d\rho} (x_N )}  \hfill \\
   &\ge & \int\limits_{\mathbb{R}^{3N} } {\frac{{\alpha _N }}
{N}\left( {\sum\limits_{i = 1}^N {|x_i |} } \right)\operatorname{d\rho}(x_1 )...\operatorname{d\rho} (x_N )}  = \alpha _N \int\limits_{\mathbb{R}^3 } {|x|\d\rho (x)} . 
\eqq
Thus $\alpha_N\le \beta$ for all $N\ge 2$.

Let us prove a lower bound. For the reader's convenience, we give now a simple bound which is enough to get that $\alpha_N$ converges to $\beta$. We will provide a better lower bound in the next step. 

Let $\{x_i\}_{i=1}^N$ be $N$ arbitrary distinct points in $\R^3$ and let $r>0$. For our purpose we may assume that $\sum_{i=1}^N |x_i|=N$. For every $i$, let $d\mu_i$ be the uniform measure on the sphere $|x-x_i|=r_i$ with the radius $r_i:=r|x_i|$ such that $\int d\mu_i=1$. Define the probability measure $\d\rho(x):=\frac{1}{N}\sum_{i=1}^N \d\mu_i(x)$. 

Since $\int |x|\d\rho(x)\ge 1$ (due to the convexity $\int |x|\d\mu_i(x) \ge |x_i|$), we have
\[
\iint\limits_{\mathbb{R}^3  \times \mathbb{R}^3 } {\frac{{x^2 \operatorname{d\rho} (x)\operatorname{d\rho}(y)}}
{{|x - y|}}}\ge \beta.
\]

On the other hand, 
\bqq
\iint\limits_{\mathbb{R}^3  \times \mathbb{R}^3 } {\frac{{x^2 \operatorname{d\rho} (x)\operatorname{d\rho} (y)}}
{{|x - y|}}} &=& N^{-2}\sum\limits_{i,j} {\iint {\frac{{x^2 \operatorname{d\mu_i} (x)\operatorname{d\mu_j} (y)}}
{{|x - y|}}}}\hfill\\
&\leqslant& N^{-2}(1 + r)^2  \sum\limits_{i,j} {\iint {\frac{{x_i^2 \operatorname{d\mu_i} (x)\operatorname{d\mu_j} (y)}}
{{|x - y|}}}}\hfill\\
   &\leqslant& N^{-2}(1 + r)^2 \left[ {\sum\limits_{i \ne j} {\frac{{x_i^2 }}
{{|x_i  - x_j |}}}  + \frac{N}{r} } \right] .
\eqq
The first inequality follows from $|x|\le (1+r)|x_i|$ for every $x$ on the sphere $|x-x_i|=r_i$, and the second inequality is due to Newton's theorem (see, e.g. \cite{LS09}, p. 91). Thus
\[
\sum\limits_{i \ne j} {\frac{{x_i^2 }}
{{|x_i  - x_j |}}}  \geqslant (1 + r)^{ - 2} N^2 \beta - r^{-1}N.
\]
This implies that 
\bq \label{eq:first-estimate-alphaN}
\alpha _N  \geqslant \frac{N}
{{N - 1}}\left[ {(1+r)^{-2}\beta  - (rN)^{-1}}\right]~~\text{for all}~r>0.
\eq
We can choose, for example, $r=N^{-1/3}$ to conclude that $\alpha_N\to \beta$ as $N\to \infty$. This ends the proof of Step 1. 
\end{proof}

We now improve the lower bound (\ref{eq:first-estimate-alphaN}).

\begin{step}\label{st:thm3-st2} We have the lower bound
$$\alpha _N  \geqslant \frac{N}{N-1}[\beta  - 3(\beta /6)^{1/3} N^{ - 2/3} ]$$
\end{step}

\begin{proof} In fact, we shall prove that 
\bq \label{eq:second-estimate-alphaN}
\alpha _N  &\geqslant& \frac{N}
{{N - 1}}\left[ {\frac{1+r^2/3}{1+r^2}\beta  - \frac{1}{rN}}\right]\ge \frac{N}
{{N - 1}}\left[ {\beta -\frac{2r^2}{3}\beta  - \frac{1}{rN}}\right]
\eq
for all $r\in (0,1]$. The desired result follows by choosing $r=(4\beta N/3)^{-1/3}$ which maximizes the right hand side of (\ref{eq:second-estimate-alphaN}). 

The bound (\ref{eq:second-estimate-alphaN}) is shown by following the same method as for (\ref{eq:first-estimate-alphaN}), but with more careful computations. We shall prove that (with the notation of the proof of Step \ref{st:thm3-st1}), for $r\in (0,1]$, 
\bq \label{eq:improve-int-xp}
\int\limits_{\mathbb{R}^3 } {|x|\d\rho (x)}  = 1 + \frac{r^2}{3}
\eq
and 
\bq \label{eq:improve-int-int-xpp}
\iint\limits_{\mathbb{R}^3  \times \mathbb{R}^3 } {\frac{{x^2 \operatorname{d\rho} (x)\operatorname{d\rho} (y)}}
{{|x - y|}}} \le N^{-2}(1 + r^2)\left[ {\sum\limits_{i \ne j} {\frac{{x_i^2 }}
{{|x_i  - x_j |}}}  + \frac{N}{r} } \right] .
\eq

The identity (\ref{eq:improve-int-xp}) follows from a direct computation using the formula
\bqq
  \int\limits_{\mathbb{R}^3 } {f(x)\d\mu _i (x)}  &=& \frac{1}
{{|S^2 |}}\int\limits_{S^2 } {f(x_i  + r_i \omega ) d\omega}  \hfill \\
   &=& \frac{1}
{{|S^2 |}}\int\limits_0^{2\pi } {\int\limits_0^\pi  {f(x_i  + r_i (\cos \theta ,\sin \theta \cos \varphi ,\sin \theta \sin \varphi ))\sin(\theta)\d\theta \, \d\varphi } } 
\eqq
for any integrable function $f$. Here the second identity comes from the spherical coordinates $\omega=(\cos \theta ,\sin \theta \cos \varphi ,\sin \theta \sin \varphi )$, where $\theta\in [0,\pi)$ and $\varphi\in [0,2\pi]$. (Note that if $r>1$, then the left hand side of (\ref{eq:improve-int-xp}) becomes $Nr(1+1/(3r^2))$.) 

Now we prove (\ref{eq:improve-int-int-xpp}). Using Newton's theorem we have 
\bqq
\iint\limits_{\mathbb{R}^3  \times \mathbb{R}^3 } {\frac{{x^2 \operatorname{d\rho} (x)\operatorname{d\rho} (y)}}
{{|x - y|}} }&=& N^{-2}\sum\limits_{i,j} {\iint\limits_{\mathbb{R}^3  \times \mathbb{R}^3 } {\frac{{x^2 \operatorname{d\mu_i} (x)\operatorname{d\mu_j} (y)}}
{{|x - y|}}}}  \hfill \\
   &\leqslant& N^{-2} \sum\limits_{i,j} {\iint\limits_{\mathbb{R}^3  \times \mathbb{R}^3 } {\frac{{x^2 \d\mu _i (x)}}
{{|x - x_j |}}}}  = \sum\limits_{i,j} {\left[ {\iint\limits_{\mathbb{R}^3  \times \mathbb{R}^3 } {\frac{{(1 + r^2 )x_i^2 \d\mu _i (x)}}
{{|x - x_j |}}} + V_{ij} } \right]}  \hfill \\
   &\leqslant& N^{-2} (1 + r^2 )\left[ {\sum\limits_{i \ne j} {\frac{{x_i^2 }}
{{|x_i  - x_j |}} + } \frac{N}
{r}} \right] + N^{-2} \sum\limits_{i,j} {V_{ij} }
\eqq
where $V_{ii}=0$ and
\bqq
V_{ij}&=&\int\limits_{\mathbb{R}^3 } {\frac{{2x_i .(x - x_i )\operatorname{d\mu_i} (x)}}
{{|x - x_j |}}}  = \frac{1}
{{|S^2 |}}\int\limits_{S^2 } {\frac{{2x_i .r_i \omega }}
{{|x_i  - x_j  + r_i \omega |}}\operatorname{d\omega} } \hfill\\
&=&- \frac{2}
{3}\frac{{r_i x_i (x_i  - x_j )\min \{ |x_i  - x_j |,r_i \} }}
{{|x_i  - x_j |(\max \{ |x_i  - x_j |,r_i \} )^2 }}~~\text{if}~i\ne j.
\eqq
Here we have used the formula
\bq \label{eq:formula-x.y/|x-y|}
\frac{1}
{{|S^2 |}}\int\limits_{S^2 } {\frac{\omega }
{{|a + s\omega |}}d\omega }  =  - \frac{1}
{3}\frac{a}
{{|a|}}\frac{{\min \{ |a|,s\} }}
{{(\max \{ |a|,s\} )^2 }},~~a\in \R^3, s>0.
\eq

Thus (\ref{eq:improve-int-int-xpp}) will be validated if we can show that $V_{ij}+V_{ji}\ge 0$. We distinguish three cases.

{\bf Case 1:} $|x_i-x_j|\ge \max\{r_i,r_j\}$. We have
\bqq
  V_{ij}  + V_{ji}  &=&  - \frac{2}
{3}\frac{{r_i^2 x_i (x_i  - x_j )}}
{{|x_i  - x_j |^3 }} - \frac{2}
{3}\frac{{r_j^2 x_j (x_j  - x_i )}}
{{|x_i  - x_j |^3 }} \hfill \\
      &=&  - \frac{{r^2 }}
{3}\frac{{(x_i^2  - x_j^2 )^2  + (x_i^2  + x_j^2 )(x_i  - x_j )^2 }}
{{|x_i  - x_j |^3 }} \leqslant  0.
\eqq

{\bf Case 2:} $|x_i-x_j|\le \min \{r_i,r_j\}$. In this case
$$
V_{ij}  + V_{ji} =  - \frac{2}
{3}\frac{{x_i (x_i  - x_j )}}
{{r_i }} - \frac{2}
{3}\frac{{x_j (x_j  - x_i )}}
{{r_j }}=-\frac{2}{3}\frac{|x_i|+|x_j|}{r}\left( {1-\frac{x_ix_j}{|x_i|.|x_j|}} \right)\leqslant 0.
$$

{\bf Case 3:} $r_i\le |x_i-x_j|\le r_j$ (the case $r_j\le |x_i-x_j|\le r_i$ is similar). We have 
\bqq
  V_{ij}  + V_{ji}  =  - \frac{2}
{3}\frac{{r_i^2 x_i (x_i  - x_j )}}
{{|x_i  - x_j |^3 }} - \frac{2}
{3}\frac{{x_j (x_j  - x_i )}}
{{r_j }} . 
\eqq
It is obvious that $V_{ji}\le 0$ since $|x_j|\ge |x_i|$. If $V_{ij}\le 0$ then we are done; if $V_{ij}\ge 0$, then using $r_i\le |x_i-x_j|$ we get
$$
V_{ij} \le  - \frac{2}
{3}\frac{{x_i (x_i  - x_j )}}
{r_i} .
$$
It turns out that $V_{ij}+V_{ji}\ge 0$ as in Case 2. 
\end{proof}

We now turn to direct bounds on $\beta$. 

\begin{step}\label{st:thm3-st3} We have the bound $0.8218\le \beta \le 0.8705$.
\end{step}
\begin{proof} The lower bound follows from the following estimate whose proof will be provided later. 
\begin{lemma} \label{le:lower-bound-(xx+yy)/d}For any positive measure $\rho$ on $\R^3$ we have
\bqq
  &~&\iint\limits_{\mathbb{R}^3  \times \mathbb{R}^3 } {\frac{{x^2  + y^2 }}
{{|x - y|}} \operatorname{d \rho} (x)\operatorname{d\rho}(y)}\hfill\\
 &\ge& \max \left\{ {~~\iint\limits_{\mathbb{R}^3  \times \mathbb{R}^3 } {\left( {\max \{ |x|,|y|\}  + \frac{{(\min \{ |x|,|y|\} )^2 }}
{{|x - y|}}} \right)\operatorname{d \rho}(x)\operatorname{d\rho}(y)}}, \right. \hfill \\
   &~&~~~~~~~~~~ \left. {\iint\limits_{\mathbb{R}^3  \times \mathbb{R}^3 } {\left( {|x - y| + \frac{2}
{3}\frac{{(\min \{ |x|,|y|\} )^2 }}
{{\max \{ |x|,|y|\} }}} \right)\operatorname{d\rho}(x)\operatorname{d \rho}(y)}} \right\}.
\eqq
\end{lemma}
\begin{remark} If $\rho$ is radially symmetric, then three terms in Lemma \ref{le:lower-bound-(xx+yy)/d} are equal.
\end{remark} 
It follows from Lemma \ref{le:lower-bound-(xx+yy)/d} that for any positive measure $\rho$ on $\R^3$ and for any $\lambda\in [0,1]$ we have 
\[
\iint\limits_{\mathbb{R}^3  \times \mathbb{R}^3 } {\frac{{x^2  + y^2 }}
{{|x - y|}}\operatorname{d \rho}(x)\operatorname{d\rho}(y)} \geqslant \iint\limits_{\mathbb{R}^3  \times \mathbb{R}^3 } {W_\lambda  (x,y)\operatorname{d\rho}(x)\operatorname{d \rho}(y)}
\]
where
\bq \label{eq:def-Wlambda}
W_\lambda  (x,y) &:=& \lambda \left( {\max \{ |x|,|y|\}  + \frac{{(\min \{ |x|,|y|\} )^2 }}
{{|x - y|}}} \right) \nn\hfill\\
&~&~+ (1 - \lambda )\left( {|x - y| + \frac{2}
{3}\frac{{(\min \{ |x|,|y|\} )^2 }}
{{\max \{ |x|,|y|\} }}} \right).
\eq
It turns out that
$$
\beta  \geqslant\mathop {\sup }\limits_{\lambda  \in [0,1]} \mathop {\inf }\limits_{x,y \in \mathbb{R}^3 } \frac{{W_\lambda  (x,y)}}
{{|x| + |y|}}.
$$
Thus the lower bound on $\beta$ follows from the following lemma whose proof is provided in the  Appendix.

\begin{lemma}\label{le:lemma7} With $W_\lambda$ being defined in (\ref{eq:def-Wlambda}) one has
$$\mathop {\sup }\limits_{\lambda  \in [0,1]} \mathop {\inf }\limits_{x,y \in \mathbb{R}^3 } \frac{{W_\lambda  (x,y)}}
{{|x| + |y|}} > 0.8218.$$
\end{lemma}
(A numerical computation shows that left-hand side of the inequality in Lemma \ref{le:lemma7} is equal to $0.8218066...$)

The upper bound on $\beta$ is attained by choosing some explicit trial measure $\rho$. By restricting $\rho$ to radially symmetric measures we have 
\[
\beta\le \beta _{\text{rad}} : = \inf \left\{ {\frac{{\int\limits_0^\infty  {\int\limits_0^\infty  {\frac{{r^2 {\rm d}m(r){\rm d}m(s)}}
{{\max \{ r,s\} }}} } }}
{{\int\limits_0^\infty  {r {\rm d}m(r)}  }}:m {\text{ a probability measure on}~ [0,\infty)}} \right\}.
\]
Choosing $m(r)=\frac{3}{4} r^{-3/2}1_{[1,9]}(r)dr$, with $dr$ being the Lebesgue measure, we get
\[
\beta _{\operatorname{rad} }  \leqslant \frac{{115}}
{{81}} - \frac{1}
{2}\ln (3) = 0.8704...
\]
(A numerical computation shows that $\beta _{\text{rad}}$ is approximately $0.8702$.)
\end{proof}

For completeness, we prove Lemma \ref{le:lower-bound-(xx+yy)/d}.

\begin{proof}[Proof of Lemma \ref{le:lower-bound-(xx+yy)/d}] We start by proving 
\bq \label{eq:(xx+yy)/d>=M+mm/d}
\iint{\frac{{x^2  + y^2 }}
{{|x - y|}}\operatorname{d \rho} (x)\operatorname{d \rho} (y)} \geqslant \iint{\left( {\max \{ |x|,|y|\}  + \frac{{(\min \{ |x|,|y|\} )^2 }}
{{|x - y|}}} \right)\operatorname{d \rho} (x)\operatorname{d \rho} (y)}.
\eq

We first show that (\ref{eq:(xx+yy)/d>=M+mm/d}) follows from the following inequality: for any $\eps>0$, if $N$ is large enough, then 
\bq \label{eq:(xx+yy)/d>=M+mm/d-discrete}
\sum\limits_{1 \leqslant i < j \leqslant N} {\frac{{x_i^2  + x_j^2 }}
{{|x_i  - x_j |}}}  \geqslant (1 - \varepsilon )\sum\limits_{1 \leqslant i < j \leqslant N} {\left( {\max \{ |x_i |,|x_j |\}  + \frac{{\min \{ |x_i |,|x_j |\} }}
{{|x_i  - x_j |}}} \right)} 
\eq
for every $\{x_i\}_{i=1}^N\subset \R^3$. In fact, we may assume that $\rho(\R^3)=1$. For every $\eps>0$, taking $N$ large enough and using (\ref{eq:(xx+yy)/d>=M+mm/d-discrete}) one has
\bqq
&~&\iint\limits_{\mathbb{R}^3  \times \mathbb{R}^3 } {\frac{{x^2  + y^2 }}
{{|x - y|}}\operatorname{d} \rho (x)\operatorname{d} \rho (y)}\hfill\\
&=&\int\limits_{\mathbb{R}^{3N} } {\frac{2}
{{N(N - 1)}}\left( {\sum\limits_{1 \leqslant i < j \leqslant N} {\frac{{x_i^2  + x_j^2 }}
{{|x_i  - x_j |}}} } \right)\operatorname{d \rho} (x_1 )...\operatorname{d \rho} (x_N )}  \hfill \\
   &\geqslant& \int\limits_{\mathbb{R}^{3N} } {\frac{{2(1 - \varepsilon )}}
{{N(N - 1)}}\sum\limits_{1 \leqslant i < j \leqslant N} {\left( {\max \{ |x_i |,|x_j |\}  + \frac{{\min \{ |x_i |,|x_j |\} }}
{{|x_i  - x_j |}}} \right)} \operatorname{d \rho} (x_1 )...\operatorname{d \rho} (x_N )}  \hfill \\
   &=& (1 - \varepsilon )\iint\limits_{\mathbb{R}^3  \times \mathbb{R}^3 } {\left( {\max \{ |x|,|y|\}  + \frac{{\min \{ |x|,|y|\} }}
{{|x - y|}}} \right)\operatorname{d \rho} (x)\operatorname{d\rho} (y)}.
\eqq
Since the latter inequality holds for every $\eps>0$, the inequality (\ref{eq:(xx+yy)/d>=M+mm/d}) follows. 

Now we show (\ref{eq:(xx+yy)/d>=M+mm/d-discrete}). This inequality follows from a key result of \cite{LSST88}. It was shown in \cite{LSST88} (Theorem 3.1) that, for any $\eps>0$, if $N$ is large enough, then 
\bq \label{eq:key-estimate-LSST88}
\mathop {\max }\limits_{1\le j\le N} \left\{ {\sum\limits_{1 \leqslant i \leqslant N,i \ne j} {\frac{1}
{{|x_i  - x_j |}}}  - \frac{{N(1 - \varepsilon )}}
{{|x_j |}}} \right\} \ge 0
\eq
for any $\{x_i\}_{i=1}^N \subset \R^3$. Since 
\[
\max \{ |x_i |,|x_j |\}  + \frac{{(\min \{ |x_i |,|x_j |\} )^2 }}
{{|x_i  - x_j |}} \leqslant \min \{ |x_i |,|x_j |\}  + \frac{{(\max \{ |x_i |,|x_j |\} )^2 }}
{{|x_i  - x_j |}},
\]
we can deduce from (\ref{eq:key-estimate-LSST88}) that
\bq \label{eq:key-estimate-LSST88-b}
\mathop {\max }\limits_{1 \leqslant j \leqslant N} \left\{ {\sum\limits_{i \ne j} {\left[ {\frac{{x_i^2  + x_j^2 }}
{{|x_i  - x_j |}} - (1 - \varepsilon )\left( {\max \{ |x_i |,|x_j |\}  + \frac{{(\min \{ |x_i |,|x_j |\} )^2 }}
{{|x_i  - x_j |}}} \right)} \right]} } \right\}\ge 0. 
\eq

Now take $1>\eps>0$. For $N$ large enough, employing (\ref{eq:key-estimate-LSST88-b}) repeatedly, we can assume that 
\[
\sum\limits_{1 \leqslant i < j} {\left[ {\frac{{x_i^2  + x_j^2 }}
{{|x_i  - x_j |}} - (1 - \varepsilon )\left( {\max \{ |x_i |,|x_j |\}  + \frac{{(\min \{ |x_i |,|x_j |\} )^2 }}
{{|x_i  - x_j |}}} \right)} \right]}  \geqslant 0
\]
for every $\eps N\le j\le N$. It turns out that
\[
\begin{gathered}
  \sum\limits_{1 \leqslant i < j \leqslant N} {\left[ {\frac{{x_i^2  + x_j^2 }}
{{|x_i  - x_j |}} - (1 - \varepsilon )\left( {\max \{ |x_i |,|x_j |\}  + \frac{{(\min \{ |x_i |,|x_j |\} )^2 }}
{{|x_i  - x_j |}}} \right)} \right]}  \hfill \\
   \geqslant \sum\limits_{1 \leqslant i < j < \varepsilon N} {\left[ {\frac{{x_i^2  + x_j^2 }}
{{|x_i  - x_j |}} - (1 - \varepsilon )\left( {\max \{ |x_i |,|x_j |\}  + \frac{{(\min \{ |x_i |,|x_j |\} )^2 }}
{{|x_i  - x_j |}}} \right)} \right]}  \hfill \\
   \geqslant  - \sum\limits_{1 \leqslant i < j < \varepsilon N} {\max \{ |x_i |,|x_j |\} }  \geqslant  - \varepsilon N\sum\limits_{1 \leqslant i < \varepsilon N} {|x_i |}  \geqslant  - \frac{\varepsilon }
{{1 - \varepsilon }}\sum\limits_{1 \leqslant i < \varepsilon N \leqslant j \leqslant N} {|x_i |}  \hfill \\
   \geqslant  - \frac{\varepsilon }
{{1 - \varepsilon }}\sum\limits_{1 \leqslant i < j \leqslant N} {\left( {\max \{ |x_i |,|x_j |\}  + \frac{{(\min \{ |x_i |,|x_j |\} )^2 }}
{{|x_i  - x_j |}}} \right)}.  \hfill \\ 
\end{gathered} 
\]
Thus (\ref{eq:(xx+yy)/d>=M+mm/d-discrete}), and hence (\ref{eq:(xx+yy)/d>=M+mm/d}), follows.  

Next, we show that 
\bq \label{eq:x.y/|x-y|>=1/3*minmin/max}
\iint\limits_{\mathbb{R}^3  \times \mathbb{R}^3 } {\frac{{x^2  + y^2 }}
{{|x - y|}}\operatorname{d} \rho (x)\operatorname{d \rho} (y)} \geqslant \iint\limits_{\mathbb{R}^3  \times \mathbb{R}^3 } {\left( {|x - y| + \frac{2}
{3}\frac{{(\min \{ |x|,|y|\} )^2 }}
{{\max \{ |x|,|y|\} }}} \right)\operatorname{d \rho} (x)\operatorname{d \rho} (y)}.
\eq
This is equivalent to
\bq \label{eq:x.y/|x-y|>=1/3*minmin/max-b}
\iint\limits_{\mathbb{R}^3  \times \mathbb{R}^3 } {\frac{{x \cdot y}}
{{|x - y|}}}\operatorname{d \rho} (x)\operatorname{d \rho} (y) \geqslant \frac{1}
{3}\iint\limits_{\mathbb{R}^3  \times \mathbb{R}^3 } {\frac{{(\min \{ |x|,|y|\} )^2 }}
{{\max \{ |x|,|y|\} }}\operatorname{d \rho} (x)\operatorname{d \rho} (y)}.
\eq
In fact, if $\rho$ is radially symmetric, then (\ref{eq:x.y/|x-y|>=1/3*minmin/max}) becomes an equality due to (\ref{eq:formula-x.y/|x-y|}). In the general case, let us introduce the positive,  radially symmetric measure
\[
\widetilde\rho (x) = \int\limits_{\operatorname{SO}(3)} {\rho (Rx)\d R}, 
\]
$\d R$ being the normalized Haar measure on the rotation group $\operatorname{SO}(3)$. Because of the positive-definiteness of the operator with the kernel $\frac{{x \cdot y}}
{{|x - y|}}$, we can employ the convexity to get 
\bqq
 \iint\limits_{\mathbb{R}^3  \times \mathbb{R}^3 } {\frac{{x \cdot y}}
{{|x - y|}}\operatorname{d \rho} (x)\operatorname{d \rho} (y)} &\ge &  \iint\limits_{\mathbb{R}^3  \times \mathbb{R}^3 } {\frac{{x \cdot y}}
{{|x - y|}}\operatorname{d \widetilde\rho} (x)\operatorname{d \widetilde\rho} (y)} \hfill \\
  &=&\frac{1}
{3}\iint\limits_{\mathbb{R}^3  \times \mathbb{R}^3 } {\frac{{(\min \{ |x|,|y|\} )^2 }}
{{\max \{ |x|,|y|\} }}\operatorname{d \widetilde\rho} (x)\operatorname{d \widetilde\rho} (y)} 
  \hfill\\
& = &\frac{1}
{3}\iint\limits_{\mathbb{R}^3  \times \mathbb{R}^3 } {\frac{{(\min \{ |x|,|y|\} )^2 }}
{{\max \{ |x|,|y|\} }}\operatorname{d\rho} (x)\operatorname{d \rho} (y)}.
\eqq
Thus (\ref{eq:x.y/|x-y|>=1/3*minmin/max-b}) (and hence (\ref{eq:x.y/|x-y|>=1/3*minmin/max})) holds for all positive measures $\rho$.
\end{proof}

\begin{appendix}
\text{}\\
{\bf \large Appendix: Technical lemmas}

In this appendix we provide the proofs of some technical lemmas.

\begin{proof}[Proof of Lemma \ref{le:lemma3}] Let us denote $\beta_1:=3(\beta /6)^{1/3}$ for short. If the desired inequality fails, then $\beta^{-1}<N/Z< 7/3$ and 
\bqq
0 &\leqslant& 1 +0.68~N^{-2/3}- (\beta ^{ - 1}  + 3Z^{-2/3})(\beta- \beta_1 N^{-2/3}) \hfill\\
&~&~+\beta _1N^{-2/3}(N/Z-\beta ^{ - 1}-3Z^{-2/3})\hfill\\
&=& N^{-2/3}\left[ {0.68- 3\beta (N/Z)^{2/3}+ \beta _1(N/Z)} \right].
\eqq

Thus the polynomial 
\[
h(x):=0.68-3\beta x^2+ \beta_1 x^3
\]
satisfies that $h(-\infty)=-\infty$, $h(0)=0.68>0$, $h(\beta^{-1/3})<0$, $h((N/Z)^{1/3})\ge 0$, $h((7/3)^{1/3})<0$ and $h(+\infty)=+\infty$ (to verify that $h(\beta^{-1/3})<0$ and $h((7/3)^{1/3})<0$ we need to use $\beta\ge 0.8218$). This implies that $h(x)$ has more than three distinct roots, which is a contradiction.
\end{proof}

\begin{proof}[Proof of Lemma \ref{le:lemma4}] Denote $\beta_1:=3(\beta/6)^{1/3}$. Assume that the desired inequality fails, namely  
\bq \label{eq:magnetic-alphaN-contradiction}
3\beta Z^{ - 2/3}  \leqslant \beta _1 N^{ - 2/3} (\beta ^{ - 1}  + 3Z^{ - 2/3} ).
\eq

Replacing the term $\beta ^{ - 1}  + 3Z^{ - 2/3}$ in the right hand side of (\ref{eq:magnetic-alphaN-contradiction}) we get
\[
\frac{N}
{Z} \geqslant \left( {\frac{{3\beta }}
{{\beta _1 }}} \right)^3  > 4
\]
since $\beta\ge 0.8218$. Thus $N\ge \max \{4Z, \beta^{-1}+3Z^{-2/3}\} > 4$. 

On the other hand, (\ref{eq:magnetic-alphaN-contradiction}) is equivalent to
\[
N^{ - 2/3}  \geqslant \frac{\beta }
{{\beta _1 }} - \frac{1}
{{3\beta (N/Z)^{2/3} }}.
\]
Using $\beta\ge 0.8218$ and $N/Z>4$ we have $N<4.5$. It contradicts the fact that $N$ must be an integer. 
\end{proof}

\begin{proof}[Proof of Lemma \ref{le:lemma7}] For any $x,y\in \R^3$, denote $a=\max\{|x|,|y|\}$, $b=\min\{|x|,|y|\}$ and $c=|x-y|$. Using the inequality $u^2+v^2\ge 2uv$ for $u,v\ge 0$, we find that
\bqq
  W_\lambda  (x,y) &=& \lambda \left( {a + \frac{{b^2 }}
{c}} \right) + (1 - \lambda )\left( {c + \frac{{2b^2 }}
{{3a}}} \right) \hfill \\
   &=& (\lambda  - \lambda ')a + \left( {\lambda 'a + (1 - \lambda )\frac{{2b^2 }}
{{3a}}} \right) + \left( {\lambda \frac{{b^2 }}
{c} + (1 - \lambda )c} \right) \hfill \\
   &\geqslant& (\lambda  - \lambda ')a + \left( {2\sqrt {\frac{2}
{3}\lambda '(1 - \lambda )}  + 2\sqrt {\lambda (1 - \lambda )} } \right)b
\eqq
for every $0\le \lambda'\le \lambda$. We may choose $\lambda'$ such that
\bq \label{eq:lambda'-in-term-lambda}
\lambda  - \lambda ' = 2\sqrt {\frac{2}
{3}\lambda '(1 - \lambda )}  + 2\sqrt {\lambda (1 - \lambda )} .
\eq
If $\lambda\ge 0.8$ the solution to (\ref{eq:lambda'-in-term-lambda}) is 
\[
\lambda ' = \left( {\sqrt {\frac{{\lambda  + 2}}
{3} - 2\sqrt {\lambda (1 - \lambda )} }  - \sqrt {\frac{2}
{3}(1 - \lambda )} } \right)^2 .
\]
Thus, for every $x,y\in \R^3$,
\[
\frac{{W_\lambda  (x,y)}}
{{|x| + |y|}} \geqslant g(\lambda ): = \lambda  - \lambda ' = \lambda  - \left( {\sqrt {\frac{{\lambda  + 2}}
{3} - 2\sqrt {\lambda (1 - \lambda )} }  - \sqrt {\frac{2}
{3}(1 - \lambda )} } \right)^2 .
\]
The desired lower bound comes from $g(0.843)= 0.821804...$ (A numerical computation shows that $g(\lambda)$ has a unique maximum at $\lambda_0= 0.843476...$ and $g_{\text{max}}=0.8218066...$).
\end{proof}
\end{appendix}

\end{document}